\documentclass[a4paper,10pt]{article}

\usepackage[margin=20mm,nohead]{geometry}
\usepackage{latexsym,graphics,graphicx,epsfig,color,enumerate}
\usepackage{amsthm,amsmath,amssymb}

\usepackage{t1enc}
\usepackage[latin2]{inputenc}

\usepackage[shortlabels]{enumitem}
\usepackage[linesnumbered,ruled]{algorithm2e}

\theoremstyle{plain}
\newtheorem{theorem}{Theorem}
\newtheorem{lemma}[theorem]{Lemma}
\newtheorem{corollary}[theorem]{Corollary}

\newtheorem{observation}[theorem]{Observation}
\newtheorem{claim}[theorem]{Claim}

\theoremstyle{definition}

\newtheorem{conjecture}[theorem]{Conjecture}

\title{Regular graphs are antimagic}
\author{
Krist\'of B\'erczi\thanks{MTA-ELTE Egerv\'ary Research Group, Department of Operations Research, E\"otv\"os University, Budapest, Hungary. E-mail: {\tt berkri@cs.elte.hu}.}
\and
Attila Bern\'ath\thanks{MTA-ELTE Egerv\'ary Research Group, Department of Operations Research, E\"otv\"os University, Budapest, Hungary. E-mail: {\tt bernath@cs.elte.hu}.}
\and
M\'at\'e Vizer\thanks{MTA Alfr\'ed R\'enyi Institute of Mathematics, P.O.B. 127, Budapest H-1364, Hungary. Email:
{\tt vizermate@gmail.com}}}

\begin{document}

\maketitle

% E-JC papers must include an abstract. The abstract should consist of a
% succinct statement of background followed by a listing of the
% principal new results that are to be found in the paper. The abstract
% should be informative, clear, and as complete as possible. Phrases
% like "we investigate..." or "we study..." should be kept to a minimum
% in favor of "we prove that..."  or "we show that...".  Do not
% include equation numbers, unexpanded citations (such as "[23]"), or
% any other references to things in the paper that are not defined in
% the abstract. The abstract will be distributed without the rest of the
% paper so it must be entirely self-contained.

\begin{abstract}
An undirected simple graph $G=(V,E)$ is called antimagic if there exists an injective function $f:E\rightarrow\{1,\dots,|E|\}$ such that $\sum_{e\in E(u)} f(e)\neq\sum_{e\in E(v)} f(e)$ for any pair of different nodes $u,v\in V$. In \cite{BBV15}, the authors gave a proof that regular graphs are antimagic. However, the proof of the main theorem is incorrect as one of the steps uses an invalid assumption. The aim of the present erratum is to fix the proof.

  % keywords are optional
  \bigskip\noindent \textbf{Keywords:} antimagic labelings; regular graphs 
\end{abstract}

\section{Introduction}

Throughout the note graphs are assumed to be simple. An undirected simple graph $G=(V,E)$ is called \textbf{antimagic} if there exists an injective function $f:E\rightarrow\{1,\dots,|E|\}$ such that $\sum_{e\in E(u)} f(e)\neq\sum_{e\in E(v)} f(e)$ for any pair of different nodes $u,v\in V$. 

Hartsfield and Ringel conjectured \cite{HR1990} that all connected graphs on at least 3 nodes are antimagic.  The conjecture has been verified for several classes of graphs, but it is widely open in general. Cranston et al. \cite{cranston} verified that regular graphs of odd degree are antimagic. With a slight modification of their argument, the authors gave a proof that even regular graphs are also antimagic in \cite{BBV15}. Recently, Chang, Liang, Pan and Zhu observed that the proof of the main theorem in \cite{BBV15} is incorrect: in the proof of Claim 6 (page 5), case 2 assumes that $f(e)>\ell$ for every $e\in E(v_{i-1})-E'_i$. However, this assumption does not hold for edges in $E^\sigma_i$, thus the subsequent calculations are incorrect.

The aim of the present note is to fix this issue. As the odd regular case was settled in \cite{cranston}, we concentrate on $k$ being even. Given a bipartite graph $G=(S,T;E)$, a path $P=\{uv,vw\}$ of length 2 with $u,w\in S$ is called an \textbf{$S$-link}. In \cite{Liang13}, Liang proposed the following conjecture and showed that, if it is true, the conjecture implies that $4$-regular graphs are antimagic.

\begin{conjecture} \label{conj:liang}
Let $G=(S,T;E)$ be a bipartite graph such that each node in $S$ has degree at most $4$ and each node in $T$ has degree at most $3$. Then $G$ has a matching $M$ and a family $\mathcal{P}$ of node-disjoint $S$-links such that every node $v\in T$ of degree $3$ is incident to an edge in $M\cup(\bigcup_{P\in\mathcal{P}} P)$.
\end{conjecture}

In \cite{BBV16}, the authors verified the conjecture by introducing a restricted path packing problem in bipartite graphs. Instead of simply modifying the original proof of Cranston et al. \cite{cranston}, we combine it with the idea of Liang \cite{Liang13} that already worked for the $4$-regular case.  

It is important to mention that at the same time when paper \cite{BBV15} appeared, regular graphs were proved to be antimagic by Chan et al. \cite{Liang16}. However, as our paper received several citations we felt that we should fix the problem appearing in the proof. Hence in the sequel we prove the following.

\begin{theorem}\label{thm:main}
Even regular graphs are antimagic.
\end{theorem}

\section{Preliminaries}

In order to make the note self-contained, we quickly go through the basic definitions. Recall that $k$ is assumed to be even. Moreover, the case $k=2$ is trivial while the case $k=4$ was settled in \cite{BBV16}, hence we concentrate on $k\geq 6$. 

Given an undirected graph $G=(V,E)$ and a subset of edges $F\subseteq E$, $F(v)$ denotes the set of edges in $F$ incident to node $v\in V$, and $d_F(v):=|F(v)|$ is the \textbf{degree} of $v$ in $F$. A \textbf{labeling} is an injective function $f:E\to \{1, 2, \dots,|E|\}$. Given a labeling $f$ and a subset of edges $F$, let $f(F)=\sum_{e\in F}f(e)$. A labeling is \textbf{antimagic} if $f(E(u))\ne f(E(v))$ for any pair of different nodes $u,v\in V$. A graph is said to be \textbf{antimagic} if it admits an antimagic labeling.

Let us recall the following folklore result from matching theory that will be used below.

\begin{theorem}\label{thm:max}
In a bipartite graph there exists a matching that covers every node of maximum degree.
\end{theorem}

We will also build upon the following theorem.

\begin{theorem}\label{thm:half}
Let $G=(S,T;E)$ be a bipartite graph and $T=T_1\cup T_2$ be a partition of $T$. For a set $X\subseteq S$ let $N_i(X)$ denote the neighbours of $X$ in $T_i$ ($i=1,2$). If $\lceil|N_1(X)|/2\rceil+|N_2(X)|\geq |X|$ for all $X\subseteq S$, then there exists a matching covering $S$ that covers at most $\lceil|T_1|/2\rceil$ nodes from $T_1$. 
\end{theorem}
\begin{proof}
Extend the graph by adding a set $S'$ of new nodes to $S$ with $|S'|=\lfloor|T_1|/2\rfloor$ together with a complete bipartite graph between $T_1$ and $S'$. We claim that the resulting bipartite graph has a matching covering $S\cup S'$. This would prove the theorem as deleting the newly added edges from such a matching results in a matching covering $S$ that covers at most $|T_1|-\lfloor|T_1|/2\rfloor=\lceil|T_1|/2\rceil$ nodes of $T_1$.

By Hall's theorem it is enough to show that for every set $Y\subseteq S\cup S'$, $|N(Y)|\geq|Y|$ holds where $N(Y)$ denote the neighbours of $Y$. It suffices to verify the inequality for $Y$'s satisfying either $Y\subseteq S$ or $S'\subseteq Y$. Indeed, if $Y\cap S'\neq\emptyset$ then for $Y'=Y\cup S'$ we have $N(Y')=N(Y)$ and $|Y'|\geq |Y|$, thus giving a more strict constraint.

If $Y\subseteq S$, then the inequality holds by the assumptions of the theorem. If $S'\subseteq Y$, then $Y=S'\cup X$ for some $X\subseteq S$, and $|N(Y)|=|N(S'\cup X)|=|T_1|+|N_2(X)|=|S'|+\lceil|T_1|/2\rceil+|N_2(X)|\geq |S'|+\lceil|N_1(X)|/2\rceil+|N_2(X)|\geq |S'|+|X|=|Y|$, concluding the proof.
\end{proof}

Another tool that our proof relies on is a theorem that appeared in \cite[Corollary 9]{BBV16} in a more general form (formulated using hypergraph terminology).

%\todo[inline]{En ugy lattam a [2, Corollary 9]-ben, hogy kell az, hogy a max fok $k$.}

\begin{theorem} \label{thm:help}
Let $G=(U,W;E)$ be a bipartite graph and $k$ be a positive even integer. Assume that each node in $W$ has degree $k-1$ and  $d_G(u)\le k$ for every $u\in U$. Then there exists a family of pairwise node-disjoint stars $(w_1,U_1;F_1),\dots,(w_q,U_q;F_q)$ such that $w_i\in W$, $|U_i|$ is either even or $k-1$, and each node $u\in U$ of degree $k$ is covered by one of the stars. %that is, $u\in\bigcup_{i=1}^q U_i$.
\end{theorem}

Let $G=(U,W;E)$ be a bipartite graph. A path $P=\{u'w,wu''\}$ of length 2 with $u',u''\in U$ is called a \textbf{$U$-link}. The \textbf{center node} of the $U$-link is $w$. Based on Theorem~\ref{thm:help}, we give the following generalization of Liang's conjecture.

\begin{theorem}\label{thm:link}
Let $G=(U,W;E)$ be a bipartite graph and $k$ be a positive even integer. Assume that each node in $U$ has degree at most $k$ and each node in $W$ has degree at most $k-1$. Then $G$ has a matching $M$ and a family $\mathcal{P}$ of node-disjoint $U$-links with center nodes having degree $k-1$ such that every node $w\in W$ of degree $k-1$ is incident to an edge in $M\cup(\bigcup_{P\in\mathcal{P}} P)$.
\end{theorem}

\begin{proof}
Observe that it suffices to verify the theorem for the special case when each node in $W$ has degree exactly $k-1$ as we can simply delete nodes of degree less than $k-1$. Let $U'\subseteq U$ denote the set of nodes having degree $k$. Consider a family of stars provided by Theorem \ref{thm:help}. The union of the edges of the stars is denoted by $F=\bigcup_{i=1}^q F_i$. Let $W'$ be the set of nodes in $W$ not covered by $F$. As $d_{E-F}(u)\leq k-1$ for each $u\in U$, $W'$ can be covered by a matching $M$ disjoint from $F$, by Theorem \ref{thm:max}.

Now we trim each star either into a matching edge or into an $U$-link. If $M$ covers at most one node from $U_i$, then keep only one edge $w_iu\in F_i$ where $u$ is not covered by $M$ (such an edge exists as $|U_i|\geq 2$). If $M$ covers at least two nodes from $U_i$, then keep two edges $w_iu',w_iu''\in F_i$ where both $u'$ and $u''$ are covered by $M$. This way we get a matching and a family of $U$-links whose union together covers $W$.
\end{proof}

As a consequence, we can give a special partition of the edges of a bipartite graph.

\begin{theorem}\label{thm:part}
Let $G=(U,W;E)$ be a bipartite graph and $k$ be a positive even integer. Assume that $1\le d_G(u)\le k$ for each node $u\in U$ and each node in $W$ has degree at most $k-1$. Then $E$ can be partitioned into three pairwise disjoint parts $E=E'\cup E^\sigma\cup E^L$ satisfying the following conditions:
\begin{enumerate}[(i)]
\item each node in $U$ has degree one in $E^\sigma$, that is, $E^\sigma$ is the union of pairwise node-disjoint stars with center nodes in $W$ together covering $U$,
\item $E^L$ is the union of pairwise node-disjoint $U$-links with center nodes having degree $k-1$ in $G$,
\item $E^\sigma\cup E^L$ covers each node in $W$ of degree $k-1$.   
\end{enumerate}
\end{theorem}

\begin{proof}
Take a matching $M$ and a family $\mathcal{P}$ of node-disjoint $U$-links provided by Theorem~\ref{thm:link}. Add $M$ to $E^\sigma$, and for each node $u\in U$ not covered by $M\cup(\bigcup_{P\in\mathcal{P}} P)$ add an arbitrary edge incident on $u$ to $E^\sigma$. Let $E^L$ consist of the edges of those $U$-links in $\mathcal{P}$ whose center nodes are not covered by $E^\sigma$. Finally, set $E'=E \setminus (E^\sigma\cup E^L)$. The partition $E=E'\cup E^\sigma\cup E^L$ thus obtained satisfies the conditions of the theorem.
\end{proof}

A \textbf{trail} in a graph $G=(V,E)$ is an alternating sequence of nodes and edges $v_0, e_1, v_1,\allowbreak \dots,\allowbreak e_t, v_t$ such that $e_i$ is an
edge connecting $v_{i-1}$ and $v_i$ for $i=1, 2, \dots, t$, and the edges are all distinct (but there might be repetitions among the nodes). The trail is \textbf{open} if $v_0\ne v_t$, and \textbf{closed} otherwise. A closed trail is also known as an Eulerian trail. 
%The \textbf{endpoints} of an open trail are $v_0$ and $v_t$. 
We will say that $e_1$ and $e_t$ are the \textbf{terminal edges} of an (open or closed) trail, while $v_0$ and $v_t$ are the \textbf{terminal nodes}. The \textbf{length} of a trail is the number of edges in it. 

\begin{claim}\label{cl:help}
Given a connected graph $G=(V,E)$, let $O=\{v\in V: d_E(v) $ is odd$
\}$. If $O\ne \emptyset$, then $E$ can be partitioned into $|O|/2$ open trails.
\end{claim}
\begin{proof}
Note that $|O|$ is even.
Arrange the nodes of $O$ into pairs in an arbitrary manner and add a
new edge between the members of every pair. Take an Eulerian trail of
the resulting graph and delete the new edges to get $|O|/2$ open trails.
\end{proof}

\begin{claim}\label{cl:help2}
If each node of a connected graph $G=(V,E)$ has even degree, then $E$ is a closed trail.
\end{claim}
\begin{proof}
A closed trail containing every edge of the graph is basically an Eulerian trail. It is well known that a graph has an Eulerian trail if and only if it is connected and every node has even degree.
\end{proof}

The main advantage of Claims~\ref{cl:help} and \ref{cl:help2} is that the edge set of the graph can be partitioned into open and closed trails such that the closed trails form connected components of the graph, while at most one open trail
starts at every node of $V$.
%Indeed, there might be a trail starting at $v$ only if $v$ has odd degree in $G$ or $v$ is contained in an Eulerian component of the graph.

\begin{corollary}\label{cor:help}
Given a bipartite graph $G=(S,T;E)$, $E$ can be partitioned into trails $T_1,\dots, T_\ell$ such that $T_i$ forms a connected component of $G$ if it is closed, and the endpoints of odd trails $T_i $ and $T_j$ are different if $i\ne j$.
\end{corollary}

%%%%%%%%%%%%%%%%%%%%%%%%%%%%%%%%%%%%%%%%%%%%%%%%%%%%%%%
\section{Proof of Theorem \ref{thm:main}} \label{sec:main}

In what follows we prove that $k$-regular graphs are antimagic for $k\geq 2$. The odd regular case was previously settled in \cite{cranston}, the case $k=2$ is trivial, and the case $k=4$ was solved in \cite{BBV16}. Hence we assume that $k$ is even and is at least $6$.

Note that it suffices to prove the theorem for connected regular graphs. Let $G=(V,E)$ be a connected $k$-regular graph and let $v^*\in V$ be an arbitrary node. Denote the set of nodes at distance exactly $i$ from $v^*$ by $V_i$ and let $q$ denote the largest distance from $v^*$. We denote the edge-set of $G[V_i]$ by $E_i$. Apply Theorem~\ref{thm:part} and Corollary~\ref{cor:help} to the induced bipartite graph $G[V_{i-1}, V_i]$  with $W=V_{i-1}$ and $U=V_i$ to get a partition $E'_i,E_i^\sigma$ and $E^L_i$ together with a trail decomposition of $E'_i$ for every $i=1,\dots,q$. Note that the BFS tree we started with makes sure that there are no isolated nodes in $U$ and the degree of a node $w\in W$ is at most $k-1$ in $G[V_{i-1}, V_i]$.

We call a connected component $C$ of $E'_i$ \textbf{critical}, if $C$ is $(k-2)$-regular
%every node in $C\cap V_{i-1}$ is covered by $E_i^\sigma$, 
and every node in $C\cap V_i$ is covered by $E_i^L$. Note that a critical component forms a closed trail.

\begin{claim} \label{cl:sol}
We can assign a $V_i$-link $\{u'v,vu''\}$ to each critical component $C$ with $u'\in C\cap V_i$ in such a way that the following holds. 
\begin{enumerate}
\item Different critical components get different $V_i$-links.
\item No open trail ends in the center nodes of two different $V_i$-links assigned to critical components.
\item If $n_o$ denotes the number of odd open trails in $E'_i$, then at most $\lceil n_o/2\rceil$ of the odd open trails end in the set of center nodes of $V_i$-links assigned to critical components.
\end{enumerate}
\end{claim}
\begin{proof}
We construct a bipartite graph as follows. One of the color classes, denoted by $S$, corresponds to the critical components of $E'_i$. The other color class, denoted by $T$, corresponds to the $V_i$-links of $E^L_i$ modulo open trails, that is, if the center nodes of two $V_i$-links form the terminal nodes of the same open trail then they are represented by the same node in the bipartite graph. We add an edge between a node corresponding to a critical component $C$ and a node representing a $V_i$-link $\{u'v,vu''\}$ if $u'\in C$.

Let $T=T_1\cup T_2$ where $T_1$ corresponds to those $V_i$-links whose center nodes are terminal nodes of odd open trails. Let $X$ be a subset of the nodes representing the critical components. We claim that the assumption of Theorem~\ref{thm:half} is satisfied, that is, $\lceil|N_1(X)|2\rceil+|N_2(X)|\geq|X|$ holds. 

Recall that a critical component $C$ corresponds to $(k-2)$-regular subgraphs in which every node in $C\cap V_i$ is covered by a $V_i$-link. As $k-2\geq 4$ and a $V_i$-link uses two edges, there are at least $2|X|$ many $V_i$-links incident to the critical components in $X$. Due to the construction of the bipartite graph, some of these $V_i$-links might be represented by the same node in $T$ (if the center nodes of two $V_i$-links form the terminal nodes of the same open trail). Let $m_1$ denote the number of $V_i$-links whose center node is the terminal node of an odd open trail, and let $m_2$ be the number of the remaining ones. Then $\lceil|N_1(X)|2\rceil+|N_2(X)|\geq \lceil m_1/2\rceil+m_2/2\geq (m_1+m_2)/2\geq|X|$ as requested.

By applying Theorem~\ref{thm:half} to the bipartite graph constructed above, we get a matching which corresponds to an assignment satisfying the conditions of the theorem, concluding the proof. 
\end{proof}

$V_i$-links assigned to critical components are called \textbf{deficient}, and we will refer to their center nodes also as \textbf{deficient nodes}. The node $u'$ and edge $u'v$ appearing in Claim~\ref{cl:sol} are called the \textbf{core node} and the \textbf{core edge} of the critical component $C$, respectively. 

The \textbf{starting node} of a closed trail is defined as follows. If the trail is a critical component, then the starting node is set to be the core node of the component. If the trail is not a critical component and has a node $v\in V_i$ with $d_{E^L_i}(v)=0$, then set the starting node to be such a node. Otherwise, set the starting node to be an arbitrary node of the trail with degree at most $k-3$. 

In what follows, we state the algorithm that provides a labeling of the graph. We reserve the $|E_q|$ smallest labels for labeling $E_q$, the next $|E_q'|+|E_q^L|$ smallest labels for labeling $E'_q\cup E_q^L$, the next $|E_q^\sigma|$ smallest labels for labeling $E_q^\sigma$, the next $|E_{q-1}|$ smallest labels for labeling $E_{q-1}$, etc. We assume that we are given a trail decomposition of $E_i'$ into a set $\mathcal{T}$ of trails together with $V_i$-links assigned to critical trails as in Claim~\ref{cl:sol} for $i=1,\dots,q$. We label the edge-sets in order $$E_q\rightarrow E'_q\rightarrow E_q^L\rightarrow E_q^\sigma\rightarrow E_{q-1}\rightarrow\dots \rightarrow E_2^\sigma\rightarrow E_1\rightarrow E'_1\rightarrow E_1^L\rightarrow E_1^\sigma.$$

For $i>0$, assume that $|E^L_i|=a$, the number of critical components in $E'_i$ is $n_i$, and that the edges of $E'_i\cup E^L_i$ are labeled using the interval $[s,\ell+a]$ (that is, $|E'_i|=\ell-s+1$). We will use the intervals $[\lceil(s+\ell)/2\rceil,\lceil(s+\ell)/2\rceil+n_i-1]\cup [\ell+a-n_i+1,\ell+a]$ for labeling the deficient $V_i$-links of $E^L_i$. The edges of the non-deficient $V_i$-links are labeled by using labels from $[\ell+n_i+1,\ell+a-n_i]$ (note that $a \ge 2n_i$). The edges of the trails appearing in the decomposition of $E'_i$ are labeled by using labels from $[s,\lceil(s+\ell)/2\rceil-1]\cup[\lceil(s+\ell)/2\rceil+n_i,\ell+n_i]$ (see Figure~\ref{fig:ints}).

\begin{figure}
    \centering
    \includegraphics[width=\textwidth]{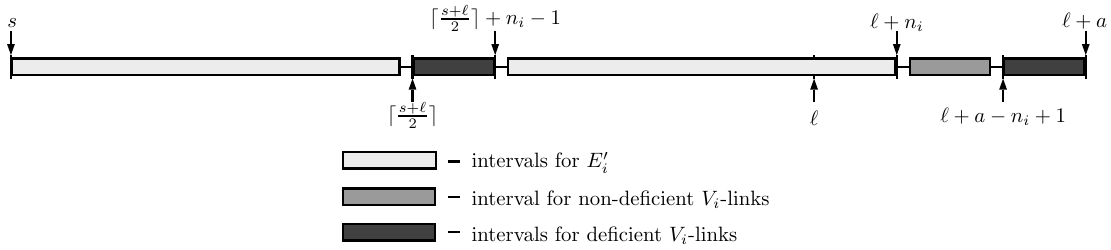}
    \caption{Assigning the intervals to $E'_i$ and $E^L_i$.}
    \label{fig:ints}
\end{figure}

\medskip

\textbf{Step 1.} Labeling the edges in $E_i$.

\noindent We label the edges of $E_i$ arbitrarily from its dedicated interval.
\medskip

\textbf{Step 2.} Labeling trails.

\newcommand{\labtr}{LabelOneTrail}

We initialize $I_1=\{s, s+1, \dots, \lceil(s+\ell)/2\rceil-1\}$
and $I_2=\{\lceil(s+\ell)/2\rceil + n_i, \lceil(s+\ell)/2\rceil +n_i + 1, \dots, \ell+n_i\}$. Notice that $|I_1|\le |I_2|\le |I_1|+1$ holds for the initial setup.
We will use the subroutine \textsf{\labtr}\ (see Algorithm \ref{alg:labtr}) for labeling one trail.
\begin{algorithm}
    \SetKwInOut{Input}{Input}
    \SetKwInOut{Output}{Output}

%    \underline{procedure \labtr} $( v_0, e_1, v_1, \dots, e_t, v_t)$\;
    \Input{A trail $T = v_0, e_1, v_1, \dots, e_t, v_t$ with a designated starting node $v_0$.}
    \Output{A labeling of $T$.}
    
  Assume that $ I_1 = \{a_1, a_1+1,\dots, b_1\}$ and $I_2= \{a_2, a_2+1, \dots, b_2\}$ are the  available intervals for labeling.

    \eIf{$v_0\in V_{i-1}$}
      {
        label $e_1, e_2, \dots, e_t$ with the labels
  $a_1, b_2, a_1+1, b_2-1, \dots$\;
      }
      {
        label $e_1, e_2, \dots, e_t$ with the labels
  $b_2, a_1, b_2-1, a_1+1, \dots$\;
      }

   Remove the used labels from $I_1$ and $I_2$.
%    \caption{Labeling one trail}\label{alg:labtr}
    \caption{\labtr $( v_0, e_1, v_1, \dots, e_t, v_t)$}\label{alg:labtr}
\end{algorithm}

When labeling the trails, we want to make sure that \textit{deficient nodes do
not get a small label}. This means the following: if $v$ is deficient then the trail $T$ that ends in $v$ will be labeled such that $v$ is the final node, and not the starting one, thus the terminal edge of $T$ at $v$ will get a label from $I_2$.
The labeling of the trails is done as follows.

\textbf{Step 2a.} While there is a not yet labeled closed trail $T= v_0, e_1, v_1, \dots, e_{2t},
v_{2t}$ with starting node $v_0$, label it by calling
\textsf{\labtr($v_0, e_1, v_1, \dots, e_{2t}, v_{2t}$)}. Notice that $|I_1|\le |I_2|\le |I_1|+1$ is maintained after this call.

\textbf{Step 2b.} While there exists a not yet labeled open even trail, take one such
trail $T= v_0, e_1, v_1, \dots, e_{2t}, v_{2t}$. By Claim \ref{cl:sol},
we can assume that $v_0$ is not deficient. Label $T$ by calling
\textsf{\labtr($v_0, e_1, v_1, \dots, e_{2t}, v_{2t}$)}. Again notice that $|I_1|\le |I_2|\le |I_1|+1$ is maintained after this call.

\textbf{Step 2c.} If all even trails are labeled then create pairs of the odd trails in
an arbitrary manner with the only restriction that at most one terminal node of the members of the pair can be deficient. This can be done since $n_n \ge n_d
- 1$ by Claim \ref{cl:sol}, where $n_d$ denotes the number of odd open
trails having a deficient terminal node, while $n_n$ denotes the number of
odd open trails having no deficient terminal node. If the number of odd
trails is odd then one trail will have no pair, and if $n_d=n_n+1$
then this trail can have a deficient terminal node. Label first the pairs
as follows. Let $T=v_0, e_1, v_1, \dots, e_{2t+1}, v_{2t+1}$ and
$T'=v'_0, e'_1, v'_1, \dots, e'_{2t'+1}, v'_{2t'+1}$ be an arbitrary
pair with $v_0\in V_i$ and $v'_0\in V_{i-1}$ where we assume that
$v'_0$ is not deficient (that is, $v_{2t+1}$ might be deficient). Call
first \textsf{\labtr($v_0, e_1, v_1, \dots, e_{2t+1}, v_{2t+1}$)} and next
\textsf{\labtr($v'_0, e'_1, v'_1, \dots, e'_{2t'+1}, v'_{2t'+1}$)} for
labeling this pair. Notice that $|I_1|\le |I_2|\le |I_1|+1$ is maintained after these  two calls.
Finally, if there is a single trail $T=v_0, e_1,
v_1, \dots, e_{2t+1}, v_{2t+1}$ that is not yet labeled then label it
by calling \textsf{\labtr($v_0, e_1, v_1, \dots, e_{2t+1}, v_{2t+1}$)}
where $v_0\in V_i$ is assumed (and $v_{2t+1}$ is either deficient or non-deficient).

\medskip

\textbf{Step 3.} Labeling deficient $V_i$-links.

\noindent Recall that deficient links are labeled  using the intervals $[\lceil(s+\ell)/2\rceil,\lceil(s+\ell)/2\rceil+n_i-1]\cup [\ell+a-n_i+1,\ell+a]$. 
In an arbitrary order, take the next deficient $V_i$-link $\{u'v,vu''\}$ and assume that the core edge is  $u'v$. Label $u'v$ with the smallest available label, and $vu''$ with the largest available label. This scheme makes sure that the sum of the labels on the link is $\lceil\frac{s+\ell}{2}\rceil + \ell+a$.
\medskip

\textbf{Step 4.} Labeling non-deficient $V_i$-links.

\noindent The edges of the non-deficient $V_i$-links are labeled by using labels from $[\ell+n_i+1,\ell+a-n_i]$ (note that $a \ge 2n_i$). 
In an arbitrary order, take the next non-deficient $V_i$-link $\{u'v,vu''\}$  and label $u'v$ with the smallest available label, and $vu''$ by the largest available label. This scheme makes sure that the sum of the labels on the link is $2\ell+a+1$.

\medskip

\textbf{Step 5.} Labeling the edges in $E_i^\sigma$.

\noindent For any node $v\in V_i$ ($i>0$), let $\sigma(v)$ denote the unique edge of $E^\sigma_i$ incident to $v$ and let $p(v)=f(E(v))-f(\sigma(v))$. Note that we have already labeled $E_q,E'_q,E_q^L, E_q^\sigma,\dots,E_{i},E'_i,\allowbreak E_i^L$, hence $p(v_i)$ is already determined for every $v_i\in V_i$. So we order the nodes of $V_i$ in an increasing order according to their $p$-value and assign the label
to their $\sigma$ edge in this order. This ensures that $f(E(u))\ne f(E(v))$ for an arbitrary pair $u,v\in V_i$.
\medskip

We have fully described the labeling procedure. This labeling scheme ensures that $f(E(v_i))< f(E(v_j))$ if $v_i\in V_i, v_j\in V_j$ and $i\ge j+2$ since $G$ is regular and the edges in $E(v_j)$ get larger labels than those in $E(v_i)$. Similarly,
$f(E(v^*))> f(E(v))$ for every $v\in V-v^*$ for the same reason.  It is only left to show that $f(E(v_i))\ne f(E(v_{i-1}))$ for arbitrary $v_i\in V_i, \  v_{i-1}\in V_{i-1}$ and $i\ge 2$.

To prove this, first we collect the observations that are true for this labeling and will be  used later. For the subsequent proofs we introduce the following notation. If $v\in V_{i-1}\cup V_i$ then let $p^L(v)=\sum_{e\in E_i^L \cap E(v)}f(e)$, $p'(v)=\sum_{e\in E'_i \cap E(v)}f(e)$ and $p(v)=\sum_{e\in E(v)-\sigma(v)}f(e)$. 

\begin{observation}\label{obs:i-1}
Let $v\in V_{i-1}$. 
\begin{enumerate}[(a)]
    \item \label{cl:succ:i-1} Successive labels on any trail incident to $v$ have sum at least $s+\ell+n_i$.
\item \label{cl:start} If $d_{E_i'}(v)$ is odd then $f(e)\ge s+n_i$ for the edge $e\in E(v)\cap E_i'$ that is the terminal edge of a trail. (This holds because we first labeled the closed trails, that includes all the critical trails.)
\item \label{cl:startdef} If $v$ is deficient (in which case $d_{E_i'}(v)=k-3$) then $f(e)\ge \lceil \frac{s+l}{2}\rceil+n_i$ for the edge $e\in E(v)\cap E_i'$ that is the terminal edge of a trail.
\end{enumerate}
\end{observation}

\begin{observation}\label{obs:i}
Let $v\in V_{i}$. 
\begin{enumerate}[(a)]
    \item \label{cl:succi} Successive labels on any trail incident to $v$ have sum at most $s+\ell+n_i$.
    \item \label{cl:fl} If $v$ is the starting node of a closed trail then the sum of the labels on the terminal edges of the trail is at most $s+\ell+n_i+(\ell-s)/2$. 
    \item \label{cl:core} If $v$ is a core node then $p^L(v)\le (s+\ell)/2 + n_i$.
\end{enumerate}
    
\end{observation}

\begin{lemma}\label{lemma:main:i-1}
For arbitrary  $v\in V_{i-1}$ and $i\geq 2$ we have $p(v)\geq(k-2)/2(s+\ell+n_i)+\ell+a$.
\end{lemma}
\begin{proof}
The idea of the proof is the following. Since $p(v)=\sum_{e\in E(v)-\sigma(v)}f(e)$ is the sum of $k-1$ edge-labels, we will pair the edges in this sum (except for one) such that the sum of the labels in each pair is $\ge s+\ell+n_i$, while the  bound $f(e)\ge \ell+a$ will be applied for the remaining edge that does not have a pair. This idea will work in almost all of the cases below.

The edges in $E_i'$ that are subsequent on a trail are naturally paired with each other by Observation \ref{obs:i-1}\ref{cl:succ:i-1}. Furthermore, if two edges both get a label $\ge \ell+a$ then they can be paired with each other.

Notice that $d_{E_i'}(v)\le k-2$ holds for $v\in V_{i-1}$.

\textbf{Case 1:} There is no $V_i$-link at $v$. Notice that the edges in $E(v)-\sigma(v)$ either fall into $E_i'$ or get a label $\ge l+a$. If $d_{E_i'}(v)=k-2$ then our rule for choosing the starting node of a closed trail will not choose $v$, that is, all edges of $E_i\cap E(v)$ are paired by the trail.
So assume that $d_{E_i'}(v)<k-2$. In this case at least two edges get a label $\ge l+a$. If $d_{E_i'}(v)$ is odd then let $e$ be the only edge at $v$ that is not paired by a trail: we will pair it with an edge that has label $\ge \ell+a$ and apply the trivial lower bound $f(e)\ge s$. If $d_{E_i'}(v)$ is even then  it is at most $k-4$, so even if $v$ is the starting node of a closed trail, the two edges $e, e'$ that are not paired by the trail (terminal edges) can be paired by edges having labels $\ge \ell+a$.

\textbf{Case 2:} There is a $V_i$-link at $v$. In this case $d_{E_i'}(v)=k-3$. If $v$ is not deficient then $p^L(v)=2l+a+1$ and $p'(v)\ge s+n_i + (k-4)/2(s+l+n_i)$, by Observation \ref{obs:i-1}\ref{cl:start}. On the other hand, if $v$ is deficient then $p^L(v)=\lceil(s+\ell)/2\rceil + \ell+a$ and $p'(v)\ge \lceil (s+\ell)/2\rceil+n_i + (k-4)/2(s+\ell+n_i)$ by Observation \ref{obs:i-1}\ref{cl:startdef}., finishing the proof.
\end{proof}

\begin{lemma}\label{lemma:main:i}
For arbitrary  $v\in V_{i}$ and $i\geq 1$, we have $p(v)\leq\frac{k-2}{2}(s+\ell+n_i)+\ell+a$.
\end{lemma}
\begin{proof}
The idea of the proof is the the same as it was in Lemma \ref{lemma:main:i-1} with the only exception that we aim for an upper bound. That is, we pair all but one of the  $k-1$ edges that appear in the formula for $p(v)$ such that the sum of the labels in each pair is $\le s+l+n_i$, while the trivial bound $f(e)\le \ell+a$ will be applied for the remaining edge that does not have a pair. 

The edges in $E_i'$ that are subsequent on a trail are naturally paired with each other by Observation \ref{obs:i}\ref{cl:succi}. Furthermore, if two edges both get a label less than $s$ then they can be paired with each other.

\textbf{Case 1:} There is no $V_i$-link at $v$. Notice that the edges in $E(v)-\sigma(v)$ either fall into $E_i'$ or get a label $ < s$. If $d_{E_i'}(v)$ is odd then there is nothing to do: we apply $f(e)\le \ell + a$ for the edge $e\in E(v)$ that is the terminal edge of a trail, and the remaining edges are either paired by the trails or have label $<s$. If $d_{E_i'}(v)$ is even then it is at most $k-2$ and there is at least one edge $h\in E(v)$ having label $<s$. If $v$ is not the starting node of a trail then all the edges at $v$ are either paired by the trails or have label $<s$. If $v$ happens to be the starting node of a closed trail then let $e$ and $e'$ be the first and the last edge of the trail and observe that $f(e)+f(h) \le s+\ell+n_i$ while we can apply the trivial bound $f(e')\le \ell+a$.

\textbf{Case 2:} There is a $V_i$-link at $v$. If $v$ is a core node then apply Observation \ref{obs:i}\ref{cl:core} to get $p^L(v)\le \frac{s+\ell}{2}+n_i$ and Observation \ref{obs:i}\ref{cl:fl} to get $p'(v) \le \frac{k-2}{2}(s+\ell+n_i)+\frac{l-s}{2}$ giving  $p(v)\le \frac{k-2}{2}(s+\ell+n_i) + \ell + n_i\le \frac{k-2}{2}(s+\ell+n_i) + \ell +a$.
If $v$ is not a core node then the trivial bound $p^L(v)\le \ell+a$ can be applied for the $V_i$-link, since $v$ is either not a starting node in a trail (in which case all edges in $E_i'\cap E(v)$ are paired by the trails and $f(e)<s$ holds for every other  edge $e\in E(v)-\sigma(v)$). On the other hand if $v$ is the starting node of a trail then either $d_{E_i'}(v)$ is odd and the terminal edge of the trail can be paired with an edge with label $<s$, or $d_{E_i'}(v)$ is even, in which case there are at least 2 edges with label $<s$: pair those with the first and the last edge of the trail.
\end{proof}

The fact that $f(\sigma(v_i))<f(\sigma(v_{i-1})$ and Lemmas~\ref{lemma:main:i-1} and \ref{lemma:main:i} together yield $f(E(v_i))< f(E(v_{i-1}))$, finishing the proof of Theorem \ref{thm:main}.

\qed 

\section*{Acknowledgement}

The authors are grateful to Chang, Liang, Pan and Zhu for pointing out the gap in the original proof.

%%%%%%%%%%%%%%%%%%%%%%%%%%%%%%%%%%%%%%%%%%%%%%%%%%%%%%%
% \bibliographystyle{plain} 
% \bibliography{myBibFile} 
% If you use BibTeX to create a bibliography
% then copy and past the contents of your .bbl file into your .tex file

\end{document}